\documentclass[10pt,journal,epsfig]{IEEEtran}

\usepackage[dvips]{graphicx}
\usepackage{graphicx}
\usepackage{amssymb}
\usepackage{cite}
\usepackage{subfigure}
\usepackage{amsmath}
\usepackage{algorithm}
\usepackage{algorithmic}
\usepackage{multirow}
\usepackage{amsthm}

\allowdisplaybreaks[4]
\usepackage{xcolor}

\begin{document}

\title{Jamming-Resistant UAV Communications: A Multichannel-Aided Approach}
\author{Bin Wang, Jun Fang,~\IEEEmembership{Senior Member,~IEEE}, Jieru Du, and Shihai Shao
\thanks{Bin Wang, Jun Fang, Jieru Du and Shihai Shao are with the National Key
Laboratory of Wireless Communications, University of
Electronic Science and Technology of China, Chengdu 611731, China,
Email: JunFang@uestc.edu.cn}
}

\maketitle

\begin{abstract}
Jamming cancellation is essential to reliable unmanned autonomous
vehicle (AAV) communications in the presence of malicious jammers.
In this paper, we develop a practical multichannel-aided
jamming cancellation method to realize secure AAV communications.
The proposed method is capable of simultaneously achieving
timing/frequency synchronization as well as jamming cancellation.
More importantly, our method does not need the signal's/jammer's
channel state information. It only utilizes
the knowledge of the legitimate sender's preamble sequence
that is available in existing communication protocols. We
also analyze the length of the preamble sequence required for
successful synchronization and signal recovery. Experimental
results on the built hardware platform show that, with a
two-antenna receiver, the proposed method can successfully
decode the signal of interest even when the jamming signal
is $40$dB stronger than the communication signal.
\end{abstract}

\begin{keywords}
AAV communication, jamming cancellation,
time and carrier frequency synchronization.
\end{keywords}


\section{Introduction}
Autonomous aerial vehicle (AAV) is a flexible, efficient
and multi-functional aircraft that has been widely used
in various tasks such as aerial photography, agriculture,
civil rescue as well as military uses. AAVs are usually
coordinated by ground or space base stations (BS). The
control instructions are transmitted from the BS to the distant
AAV through wireless links. To enable the normal
operation of AAV, it is crucial to develop techniques
maintaining reliable communications between the BS and
the AAV in the presence of strong jamming signals
\cite{LuXiao20,LvXiao23}.

The most common anti-jamming techniques include spread
spectrum techniques, such as frequency hopping spread
spectrum (FHSS) \cite{Torrieri00} and direct sequence
spread spectrum (DSSS) \cite{Flikkema97}. Nevertheless,
the spread spectrum techniques can only support low data
rate communications. In this work, we focus our study on
multichannel-assisted jamming suppression techniques.
Interference/jamming suppression based on multichannel
signal processing is a topic that has been studied for
many years, and a variety of algorithms have been proposed
\cite{LiStoica05,GollakotaAdib11,YanZeng16,MartiKolle23,
DoBjornson17,DoBjornson17a,MartiArquint24,ZengCao17,PirayeshSangdeh20,PirayeshSangdeh21}.
Among them, the most renowned technique is adaptive beamforming
(ABF) \cite{LiStoica05} which cancels the jamming signals
by adaptively forming a beam-pattern that rejects signals
from undesirable directions. Nevertheless, ABF requires the
impinging direction of the desired signal, which is difficult
to estimate in urban environments where multipath propagation
is significant. Different from ABF, some other multichannel-assisted
anti-jamming techniques were developed by utilizing the
knowledge of the legitimate channel or the pilot sequence. Specifically, the work
\cite{GollakotaAdib11,YanZeng16} proposed to first
estimate the legitimate channel during the jamming reaction
period in which no jamming signal is sent, and then,
based on the knowledge of the legitimate channel, they
extract the information of the jamming channel from
the received signal's covariance matrix. This scheme, however,
faces difficulties when there is no jamming reaction period
or the reaction period is very short. Another class of
methods \cite{MartiKolle23,DoBjornson17,DoBjornson17a,MartiArquint24},
referred to as semi-blind source separation techniques,
utilize the known pilot sequence to cancel the desired
signal component from the received signals, thus enabling
to recover the jamming channel subspace. This semi-blind
approach, however, requires perfect time and carrier synchronization
to remove the signal component, while time and carrier
synchronization itself is very challenging in the presence
of strong jamming signals. To remedy this issue, the work
\cite{ZengCao17,PirayeshSangdeh20,PirayeshSangdeh21} proposed a jamming-resilient synchronization
module to perform time and frequency synchronization. The
basic idea is to first use the minimum eigenvector (i.e.
the eigenvector associated with the minimum eigenvalue)
of the received signal's covariance matrix as a spatial
filter to suppress the jamming signals, and then apply
conventional synchronization schemes to perform time/frequency
synchronization. Such a constructed spatial filter,
however, has the tendency to suppress the desired signal
as well, thus leading to an unsatisfactory performance.

To overcome the drawbacks of existing methods, we, in our
paper, propose a novel preamble-assisted multichannel
signal processing method which can simultaneously achieve
time/frequency synchronization as well as jamming cancellation.
To our best knowledge, this is the first work that utilizes
the preamble sequence for joint time/frequency synchronization
and jamming cancellation. Unlike \cite{GollakotaAdib11,YanZeng16},
our proposed algorithm does not need to estimate the legitimate channel. It only utilizes the preamble
sequence that is periodically transmitted by the BS. Another
contribution of our work lies in that we provide a rigorous
theoretical justification for the proposed method, and analyze
the minimum length of the preamble sequence that is required
for successful synchronization and signal recovery. Experimental
results on both simulated data and universal software radio
peripheral (USRP) platform show the superiority of the proposed
method over state-of-the-art anti-jamming methods.

\section{System Model and Problem Formulation}
We consider a downlink AAV communication scenario where the
desired communication signal $s(t)$ is transmitted from a single-antenna
BS to the AAV. The signal is interfered by a number of strong
co-channel jamming signals, denoted as $\{i_k(t)\}_{k=1}^{K}$.
An $N$-antenna receiver is employed at the AAV to receive
and decode the desired signal. We assume that $N>K$. A
narrowband model is considered in this paper, in which
case the signal received by the AAV can be expressed as
\begin{align}
\textstyle\boldsymbol{y}(t)=e^{j2\pi\delta_{f}t}
\boldsymbol{h}s(t-\tau)+\sum_{k=1}^{K}\boldsymbol{g}_k i_k(t)
 + \boldsymbol{n}(t)
\label{receive-signal-model}
\end{align}
{\color{blue}where $\boldsymbol{y}(t)\in\mathbb{C}^{N}$ denotes the
received signal at time instant $t T_s$, $T_s$ is the sampling interval (which is omitted
for sake of notational convenience)}, $\boldsymbol{h}\in\mathbb{C}^{N}$
represents the channel between the BS and the AAV,
$\boldsymbol{g}_k\in\mathbb{C}^{N}$ stands for the channel
between the $k$th jammer and the AAV, $\boldsymbol{n}(t)\in\mathbb{C}^{N}$
is the additive white Gaussian noise, the term $e^{j2\pi\delta_{f}t}$
is used to characterize the carrier frequency offset (CFO) between
the BS and the AAV, and we use $s(t-\tau)$ to account for the
unknown timing offset between the BS and the AAV. {\color{blue}Note that here
$s(t)$ is the digitally modulated baseband signal such as QPSK
(quadrature phase shift keying) or QAM (quadrature amplitude
modulation).} As compared with indoor or ground environments,
the channel between the AAV and the BS is more likely to be
dominated by the line-of-sight (LoS) path component, whereas
the indoor or ground channels may consist of a large
number of multi-path components. Nevertheless, since
we do not rely on any specific structure
on the channels $\boldsymbol{h}$ and $\{\boldsymbol{g}_k\}_{k=1}^K$,
our model (\ref{receive-signal-model}) is general and
applies to both AAV and ground communication scenarios.


The following basic assumptions are adopted in this
paper:
\begin{itemize}
\item[A1] The channels $\boldsymbol{h}$ and $\{\boldsymbol{g}_k\}_{k=1}^K$
are unknown to the receiver. Besides, they are linearly
independent of each other and keep invariant within each channel coherence block.
\item[A2] The legitimate signal $s(t)$ and the jamming signals
$\{i_k(t)\}_{k=1}^{K}$ are random signals which are statistically
independent of each other.
\end{itemize}


In this work, we assume that the BS periodically sends a
preamble sequence $\{s(t), \ t=1,\ldots,T\}$ that is known
to the AAV. Note that most communication protocols,
e.g. IEEE 802.11ac or 802.11n \cite{HayatYanmaz15} which are
widely employed in AAV communications, include periodically
transmitted preamble sequences in their transmission protocol
in order to perform time and carrier frequency synchronization.

Specifically, we use the preamble sequence as a reference signal
to design a CFO-compensated spatia-temporal filter
$\{e^{-j2\pi\omega (t+l)}\boldsymbol{w}_{l}^{H}\}_{l=0}^{L-1}$
such that the output of the filter is as close to the
reference signal as possible:
\begin{align}
\textstyle \min\limits_{\omega,\{\boldsymbol{w}_l\}_{l=0}^{L-1}}
\sum\limits_{t=1}^{T}\big|s(t)-\sum\limits_{l=0}^{L-1}e^{-j2\pi\omega(t+l)
}\boldsymbol{w}_{l}^{H}\boldsymbol{y}(t+l)\big|^2 \label{opt1}
\end{align}
where $(\cdot)^{H}$ denotes the conjugate transpose of a vector
or a matrix, and $L$ is the order of the filter satisfying $L>\tau$.
Note that we need to resort to the time-varying phase term
$e^{-j2\pi\omega (t+l)}$ to compensate for the CFO. Simply using
the time-independent coefficients $\{\boldsymbol{w}_l\}_{l=0}^{L-1}$
cannot correct the CFO and thus cannot filter out the desired signal.
Without loss of generality, we assume $\{s(t), \ \forall \text{$t\leq 0$ or $t>T$}\}$
are data symbols that are unknown to the receiver.


Our objective is to design a spatia-temporal filter to
successfully suppress the jamming signals and recover
the desired communication signal.
Define
\begin{align}
&\boldsymbol{\vec{y}}_{\omega}(t) \triangleq [e^{j2\pi \omega t
}\boldsymbol{y}^H(t) \phantom{0}\ldots\phantom{0}e^{j2\pi \omega
(t+L-1) }\boldsymbol{y}^H(t+L-1)]^H
\nonumber\\
&\boldsymbol{w}\triangleq
[\boldsymbol{w}_0^H\phantom{0}\ldots\phantom{0}\boldsymbol{w}_{L-1}^H]^H\in\mathbb{C}^{NL}
\end{align}
{\color{blue}where
$\boldsymbol{w}_l$ is the $l$th tap's filter coefficients expressed into a vector form.}
The optimization problem (\ref{opt1}) can be re-expressed as
\begin{align}
\textstyle \min_{\boldsymbol{w},\omega} \
\sum_{t=1}^{T}\big|s^*(t)-\boldsymbol{\vec{y}}_{\omega}^H(t)\boldsymbol{w}\big|^2
\label{opt2}
\end{align}
where $(\cdot)^{*}$ denotes the complex conjugate of a complex
number. Problem (\ref{opt2}) can be further compactly written
as
\begin{align}
\textstyle \min_{\boldsymbol{w},\omega} \
\|\boldsymbol{\tilde{s}}-\boldsymbol{A}_{\omega}\boldsymbol{w}\|_2^2
\label{opt3}
\end{align}
where
\begin{align}
&\boldsymbol{\tilde{s}}\triangleq
[s(1)\phantom{0}\ldots\phantom{0}s(T)]^H\in\mathbb{C}^{T}
\\
&\boldsymbol{A}_{\omega}\triangleq
[\boldsymbol{\vec{y}}_{\omega}(1)\phantom{0}\ldots\phantom{0}
\boldsymbol{\vec{y}}_{\omega}(T)]^H\in\mathbb{C}^{T\times NL}
\end{align}

\section{Spatia-Temporal Filter Design}
\label{sec-ST-filter-design}
Designing a spatia-temporal filter amounts to solving problem
(\ref{opt3}). However, due to the coupling between $\omega$
and $\boldsymbol{w}$, it is difficult to solve this problem.
To address this problem, we first fix $\omega$.
Then the least squares solution of $\boldsymbol{w}$ can be
easily obtained as
\begin{align}
\boldsymbol{w}^{\star}
=(\boldsymbol{A}_{\omega}^H\boldsymbol{A}_{\omega})^{+}
\boldsymbol{A}_{\omega}^H\boldsymbol{\tilde{s}} \label{CFO-1}
\end{align}
where $(\boldsymbol{A}_{\omega}^H\boldsymbol{A}_{\omega})^{+}$ is
the pseudo inverse of $\boldsymbol{A}_{\omega}^H
\boldsymbol{A}_{\omega}$. Here we are interested in the
under-determined regime where the length of the preamble sequence is smaller
than the dimension of the filter to be designed (i.e. $T<NL$) and thus
$\boldsymbol{A}_{\omega}^H\boldsymbol{A}_{\omega}$ is
rank-deficient. When $T\geq NL$, the matrix $\boldsymbol{A}_{\omega}^H\boldsymbol{A}_{\omega}$
is very likely to be invertible, in which case we can simply
use $(\boldsymbol{A}_{\omega}^H\boldsymbol{A}_{\omega})^{-1}$.

Substituting $\boldsymbol{w}^{\star}$ back into
problem (\ref{opt3}) yields
\begin{align}
\min\nolimits_{\omega}\quad
f_c(\omega)\triangleq\|\boldsymbol{\tilde{s}}-\boldsymbol{A}_{\omega}
(\boldsymbol{A}_{\omega}^H\boldsymbol{A}_{\omega})^{+}
\boldsymbol{A}_{\omega}^H\boldsymbol{\tilde{s}}\|_2^2 \label{CFO-2}
\end{align}
Problem (\ref{CFO-2}) involves only the optimization of
$\omega$. A one-dimensional search scheme can be employed to
obtain the optimal $\omega^{\star}$. Without loss of generality,
we assume that $\omega^{\star}\in[\delta_{\text{min}},\delta_{\text{max}}]$.
We choose $m$ equidistant points in
$[\delta_{\text{min}},\delta_{\text{max}}]$, denoted as
$\{\omega^i\}_{i=1}^m$, and then compute the value of
$f_c(\omega)$ for each $\omega^i$. The optimal $\omega^{\star}$
is chosen as the one that achieves the smallest value of
$f_c(\omega)$, i.e.
\begin{align}
\omega^{\star}=\arg\min\nolimits_{\omega\in\{\omega^i\}_{i=1}^m} \ f_c(\omega)
\end{align}
Then the optimal filter $\boldsymbol{w}^{\star}=
\{\boldsymbol{w}_l^{\star}\}_{l=0}^{L-1}$ is given as
\begin{align}
\boldsymbol{w}^{\star}=(\boldsymbol{A}_{\omega^{\star}}^H
\boldsymbol{A}_{\omega^{\star}})^{+}\boldsymbol{A}_{\omega^{\star}}^H
\boldsymbol{\tilde{s}} \label{RLS-filter-1}
\end{align}

{\color{blue}We would like to clarify that our proposed method does not
need to explicitly estimate the jamming signals and subtract
them from the received signals. Instead, it uses a spatial-temporal
filter to automatically eliminate the jamming signals $\{i_k(t)\}$
and recover the desired communication signal $s(t)$. As explained
later in this paper, when a certain condition is satisfied,
the optimized filter $\boldsymbol{w}$ will become orthogonal
to the jamming channels $\{\boldsymbol{g}_k\}_{k=1}^{K}$,
i.e. $\boldsymbol{g}_k^{H}\boldsymbol{w}_l=0,\forall k,l$.
Thus, the jamming signals will be automatically removed by
the filter.}

\subsection{Theoretical Analysis}
\label{sec-theoretical-analysis}
In this subsection, we attempt to answer under what conditions the solution
to problem (\ref{opt3}) can completely cancel the jamming signals and recover
the desired signal $s(t)$.

Note that a filter which can successfully recover the desired signal
$s(t)$ has to satisfy the following conditions:
\begin{align}
\omega^{\star}=\delta_{f},  \qquad
\boldsymbol{w}^{\star}\in \mathcal{C}, \label{condition2}
\end{align}
where $\mathcal{C}$ is defined as
\begin{align}
\mathcal{C}\triangleq\{\{\boldsymbol{w}_l\}_{l=0}^{L-1} \
| \ & \boldsymbol{h}^H\boldsymbol{w}_{\tau}=1, \
\boldsymbol{h}^H\boldsymbol{w}_{l}=0, \forall l\neq \tau,
\nonumber\\
&\boldsymbol{g}_k^H\boldsymbol{w}_{l}=0, \forall l \ \forall k \}.
\label{set-C}
\end{align}
To understand condition (\ref{condition2}), note that in the noiseless case, the output of the
CFO-compensated filter can be expressed as
\begin{align}
&\textstyle\sum_{l=0}^{L-1}e^{-j2\pi \omega (t+l)}\boldsymbol{w}_{l}^{H}\boldsymbol{y}(t+l)
\nonumber\\
=&\textstyle\sum_{l=0}^{L-1}e^{-j2\pi \omega (t+l)}
\boldsymbol{w}_{l}^{H}\Big(e^{j2\pi\delta_{f}(t+l)}\boldsymbol{h} s(t+l-\tau)
\nonumber\\
&\textstyle\qquad\qquad\qquad\qquad+
\sum_{k=1}^{K}\boldsymbol{g}_k i_k(t+l)\Big)
\nonumber\\
\stackrel{(a)}=&\textstyle
\sum_{l=0}^{L-1}
\Big(\boldsymbol{w}_{l}^{H}\boldsymbol{h} s(t+l-\tau)
\nonumber\\
&\textstyle \quad\quad+e^{-j2\pi \omega
(t+l)}\sum_{k=1}^{K}\boldsymbol{w}_{l}^{H}\boldsymbol{g}_k
i_k(t+l)\Big)
\stackrel{(b)}= s(t)
\end{align}
where $(a)$ holds when $\omega=\delta_{f}$, and $(b)$ follows when $\boldsymbol{w}\in \mathcal{C}$. From the above we see that, any filter which satisfies
(\ref{condition2}) is able to successfully suppress the
jamming signals and filter out the desired signal. We have the
following result concerning the condition under which any
solution of (\ref{opt3}) satisfies (\ref{condition2}).

\newtheorem{theorem}{Theorem}
\begin{theorem}
\label{theorem-1} Assume that $N\geq K+1$, where $N$ is the
number of antennas at the receiver and $K$ is the number of
jamming signals. Let $\boldsymbol{s}\in\mathbb{C}^{T}$ be a length-$T$
preamble sequence known by the receiver. For the
noiseless case, if $T>(K+1)L$, then any solution of (\ref{opt3}) satisfies (\ref{condition2}).
\end{theorem}
\begin{proof}
See Appendix \ref{appA}.
\end{proof}


In Theorem \ref{theorem-1}, we have shown that the condition $T>(K+1)L$ guarantees to find an effective spatial-temporal filter to remove
the jamming signals and recover the desired communication signal.
Remember that the number of jamming signals $K$ is assumed to be
smaller than the number of antennas $N$. As a result, we have $(K+1)L\leq NL$.
This result implies that the length of the preamble sequence is not necessarily greater than
the dimension of the filter, which helps achieve a sample complexity reduction.

\subsection{Efficient Implementations}
In our proposed method, we need to compute the value of
$f_c(\omega)$ for every $\omega^i\in\{\omega^i\}_{i=1}^m$.
For a specific $\omega^i$, the major computational task is to
compute $(\boldsymbol{A}_{\omega^{i}}^H\boldsymbol{A}_{\omega^{i}})^{+}$.
To perform this task in an efficient way,
we alternatively consider calculating $(\epsilon\boldsymbol{I}
+\boldsymbol{A}_{\omega^i}^H\boldsymbol{A}_{\omega^i})^{-1}$,
where $\epsilon$ is set to a small positive value.

Recall that the $t$th row of $\boldsymbol{A}_{\omega^i}$
is $\boldsymbol{\vec{y}}_{\omega^i}^H(t)$, thus
$\epsilon\boldsymbol{I}+\boldsymbol{A}_{\omega^i}^H
\boldsymbol{A}_{\omega^i}$ can be equivalently written as
\begin{align}
\textstyle\epsilon\boldsymbol{I}+\boldsymbol{A}_{\omega^i}^H
\boldsymbol{A}_{\omega^i}=\epsilon\boldsymbol{I}+
\sum_{t=1}^T\boldsymbol{\vec{y}}_{\omega^i}(t)
\boldsymbol{\vec{y}}_{\omega^i}^H(t)
\end{align}
Now define $\boldsymbol{D}_{\omega^i}^0\triangleq\epsilon\boldsymbol{I}$ and
\begin{align}
&\textstyle\boldsymbol{D}_{\omega^i}^{t'}\triangleq\boldsymbol{D}_{\omega^i}^0
+\sum_{t=1}^{t'}\boldsymbol{\vec{y}}_{\omega^i}(t)
\boldsymbol{\vec{y}}_{\omega^i}^H(t), \ 1\leq t'\leq T.
\end{align}
Clearly, we have $\boldsymbol{D}_{\omega^i}^{T}=\epsilon\boldsymbol{I}+\boldsymbol{A}_{\omega^i}^H
\boldsymbol{A}_{\omega^i}$. For $\boldsymbol{D}_{\omega^i}^{t'}$,
we have
\begin{align}
&(\boldsymbol{D}_{\omega^i}^{t'})^{-1}
=\Big(\boldsymbol{D}_{\omega^i}^{t'-1}+\boldsymbol{\vec{y}}_{\omega^i}(t')
\boldsymbol{\vec{y}}_{\omega^i}^H(t')\Big)^{-1}
\nonumber\\
&\overset{(a)}{=}(\boldsymbol{D}_{\omega^i}^{t'-1})^{-1}
-(\boldsymbol{D}_{\omega^i}^{t'-1})^{-1}\boldsymbol{\vec{y}}_{\omega^i}(t')
\nonumber\\
&\Big(1+\boldsymbol{\vec{y}}_{\omega^i}^H(t)
(\boldsymbol{D}_{\omega^i}^{t'-1})^{-1}
\boldsymbol{\vec{y}}_{\omega^i}(t)\Big)^{-1}
\boldsymbol{\vec{y}}_{\omega^i}^H(t)(\boldsymbol{D}_{\omega^i}^{t'-1})^{-1}
\label{effici-imple}
\end{align}
where $(a)$ has invoked the Woodbury identity.
This means that $(\boldsymbol{D}_{\omega^i}^{T})^{-1}
=\epsilon\boldsymbol{I}+\boldsymbol{A}_{\omega^i}^H
\boldsymbol{A}_{\omega^i}$ can be obtained by recursively
performing (\ref{effici-imple}) from $t'=1$ to $t'=T$.
A prominent advantage of performing this recursion is
that this can be implemented in a streaming fashion.


The computational complexity
of this recursive method is in the order of $O(mTN^2L^2)$,
where $m$ accounts for the $m$ points in $\{\omega^i\}_{i=1}^m$,
$T$ is the length of the preamble sequence,
and $O(N^2L^2)$ is the complexity for computing the
matrix-vector product $(\boldsymbol{D}_{\omega^i}^{t'-1})^{-1}
\boldsymbol{\vec{y}}_{\omega^i}(t')$ in (\ref{effici-imple}).

\begin{figure*}
  \centering
  \includegraphics[width=16cm,height=2.3cm]{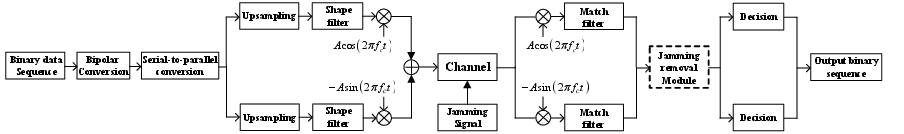}
  \caption{The flowchart of the modulation and signal processing.}
  \label{fig1}
\end{figure*}


\section{Experimental Results}
In this section, we provide experimental results to illustrate
the effectiveness of the proposed method on both simulated
data and the hardware testbed. We compare our proposed method with
the state-of-the-art two-stage spatial filtering algorithms \cite{ZengCao17,PirayeshSangdeh21}
which use a two-stage scheme to cancel the jamming signal in
the presence of timing and CFO offset. Note that the method \cite{PirayeshSangdeh21} requires to know
the channel ratio of the jamming channel, which is assume available to \cite{PirayeshSangdeh21}.

\subsection{Simulated Results}
\label{sec-simulated-experiment}
In our simulations, the number of antennas
is set to $N=4$ and the number of jamming signals is set to $K=3$.
The QPSK modulation is employed. The flow chart of the modulation
and the signal processing process is shown in Fig. \ref{fig1}. The
carrier frequency is set to $5$GHz. The shape filter and the
match filter are chosen as the raised cosine finite impulse response
filter, with their length fixed as $49$ and the roll-off factor set
to $0.5$. The CFO between the transmitter and the receiver is set
to $760$Hz. The unknown time offset $\tau$ is randomly generated
within the interval $[0.25\mu s, 2.5\mu s]$ (corresponding to
$[1,10]$ sampling points at the receiver).

The baseband signal is a $0/1$ binary sequence. In QKSP modulation,
the symbol rate is set to $0.5$MB/s, the
upsampling ratio is set to $8$ and the sampling rate at the
receiver is set to $4$MHz. We assume that each frame consists
of $164$ binary bits. The first $T$ bits are used to generate
the preamble sequence while the rest bits are data bits. Due to
the serial-to-parallel conversion as well as the upsampling
operation, each frame has a total of $656$ samples, in which the
number of preamble samples is $4T$. The jamming signals are
randomly generated according to a normal distribution and then
directly added to the received baseband signal. The T2R channel
and the J2R channel are generated as Gaussian random vectors.
The signal-to-jamming ratio (SJR) is defined as
\begin{align}
&\textstyle\text{SJR}\triangleq
\log_{10}\left(\frac{\|\boldsymbol{h}\|_2^2}{\sum_{k=1}^K\|\boldsymbol{g}_k\|_2^2}\right)
\end{align}
For our proposed method, the search range of $\omega$ is set
to $[\delta_{\text{min}},\delta_{\text{max}}]=[0,1000\text{Hz}]$,
and the number of equidistant points is set to $m=200$,
which corresponds to a search interval of $5$Hz. Also, the
order of the filter $L$ is set to $L=12$.

Fig. \ref{fig2}(a) plots the bit error rate (BER) achieved by
respective methods versus the length of the preamble sequence
$T$, where the signal-to-noise ratio (SNR) is set to $5$dB and
the SJR is set to $-30$dB. Results are averaged over $10^3$
Monte Carlo runs. We can see from Fig. \ref{fig2}(a) that our
proposed method outperforms the competing algorithms by a big
margin. Also, by increasing the length of the preamble sequence,
our proposed method is able to achieve a substantial performance
improvement. Fig. \ref{fig2}(b) plots the BER achieved by respective
methods versus the SJR under different SNRs, where we set $T=40$.
Fig. \ref{fig2}(b) shows that our proposed method attains a much
lower BER than its competing algorithms. It is observed that the
BER of the methods \cite{ZengCao17,PirayeshSangdeh21} slightly
increases when the SJR increases from $-30$dB to $-10$dB.
This is because the subspace of the receive covariance
matrix is dominated by both the jamming channel and the communication
channel when the strength of the desired signal is comparable to
that of the jamming signal. Hence the spatial filter chosen as
the minimum eigenvector of the receive covariance matrix
\cite{ZengCao17,PirayeshSangdeh21} has the tendency to suppress
the desired signal, thus leading to a deteriorated
performance.


\begin{figure}[!t]
 \centering
\begin{tabular}{cc}
\hspace*{-3ex}
\includegraphics[width=4.9cm,height=4.3cm]{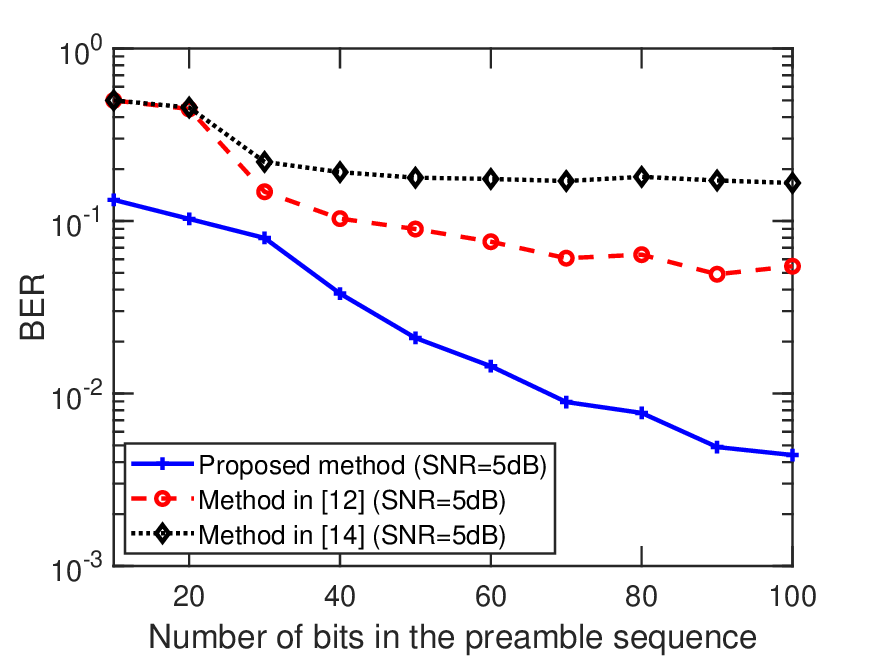}&
\hspace*{-5ex}
\includegraphics[width=4.9cm,height=4.3cm]{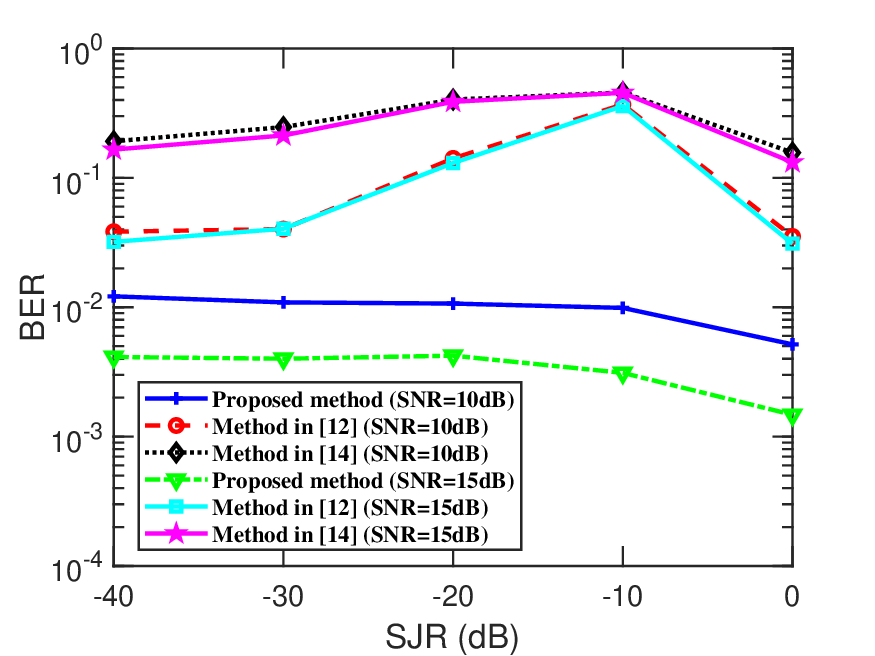}
\\
(a)& (b)
\end{tabular}
  \caption{(a). BER vs. length of the preamble sequence; (b). BER vs. SJR under different SNRs.}
   \label{fig2}
\end{figure}

\begin{figure}
  \centering
  \includegraphics[width=5.5cm,height=3cm]{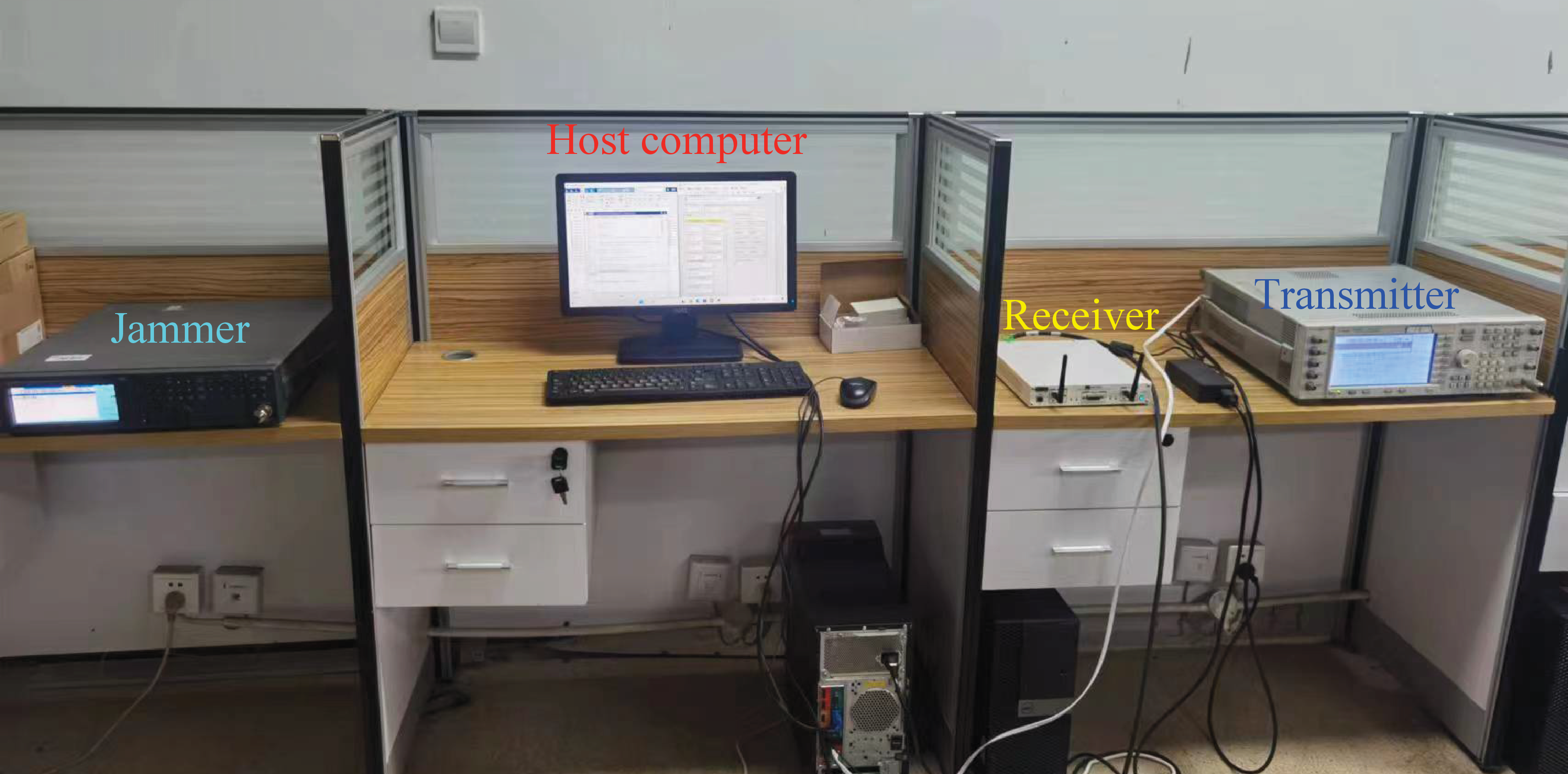}
  \caption{The testbed system. From left to right: Jammer, host computer, receiver, and transmitter.}
  \label{fig4}
\end{figure}

\begin{figure}
  \centering
  \includegraphics[width=5.5cm,height=4cm]{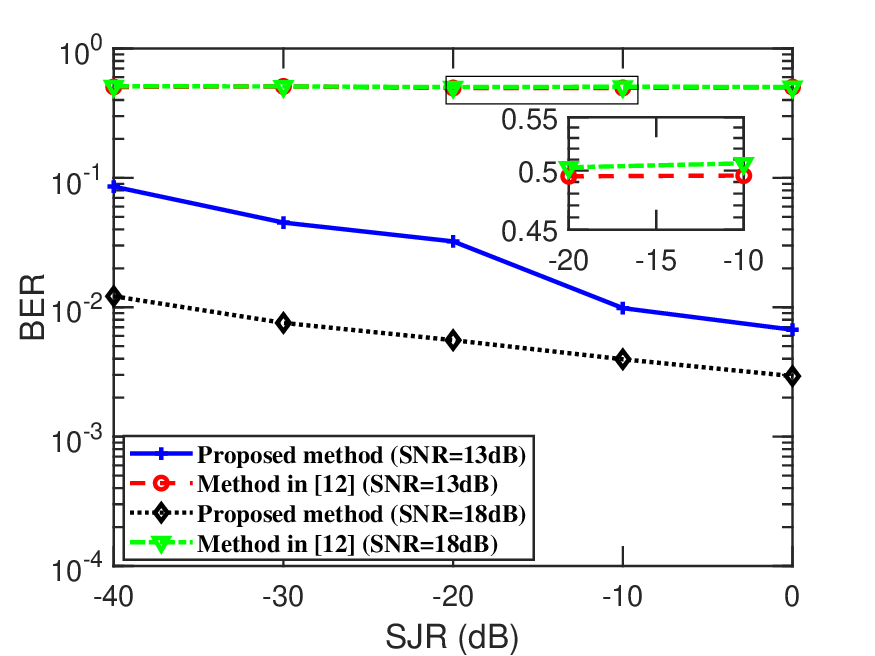}
  \caption{Testbed results: BER vs. SJR under different values of SNR.}
  \label{fig5}
\end{figure}


\subsection{Testbed Results}
We also conduct experiments on the USRP radio platform as
illustrated in Fig. \ref{fig4}. In Fig. \ref{fig4}, a vector
signal generator is used as the legitimate transmitter. The
receiver is a two-antenna USRP and an analog signal generator
is used to produce the jamming signal. Settings are the same
as those in the previous subsection, including
the QPSK modulation, carrier frequency, symbol rate, and
shape/matched filter parameters. The length of the preamble sequence is
set to $T=70$ and the frame length is fixed as $164$.
The jamming
signal is a complex exponential signal whose frequency equals
to the carrier frequency. The groundtruth CFO between the
transmitter and the receiver is approximately $760$Hz.
Fig. \ref{fig5} plots the BERs achieved by respective
methods as a function of the SJR. Since the channel ratio
of the jamming channel cannot be obtained in this experiment,
the method \cite{PirayeshSangdeh21} is not included for
comparison. From Fig. \ref{fig5} we see that our proposed
method attains a much lower BER than the two-stage spatial filtering method
\cite{ZengCao17}. Also, it is observed that our proposed method
delivers a decent BER performance even when the jamming signal's power is $40$dB stronger than the legitimate transmitter's power.

\section{Conclusions}
In this paper, we proposed a practical multi-channel-assisted
method for jamming cancellation for AAV communications. The
proposed method only utilizes the transmitter's preamble sequence to simultaneously achieve time/frequency synchronization
as well as jamming cancellation. Experimental results demonstrate the superiority of the proposed
method over state-of-the-art anti-jamming methods.


\appendices
\section{Proof of Theorem \ref{theorem-1}} \label{appA}
First we show that if $T>(K+1)L$, then
the objective function in (\ref{opt3}) attains its minimum value 0
only when $\omega=\delta_f$. Consider the following equation:
\begin{align}
\boldsymbol{A}_{\omega}\boldsymbol{w}=\boldsymbol{\tilde{s}}
\label{CFO-proof-1}
\end{align}
The $t'$th ($1\leq t'\leq T$) row of this equation can be written as
\begin{align}
\textstyle s^{*}(t')=&\textstyle \boldsymbol{\vec{y}}^{H}(t')\boldsymbol{w}=
\sum_{l=0}^{L-1}e^{j2\pi \omega (t'+l)}
\boldsymbol{y}^H(t'+l)\boldsymbol{w}_{l}
\nonumber\\
=&\textstyle\sum_{l=0}^{L-1}\big(e^{j2\pi (\omega-\delta_f) (t'+l)}\cdot
\boldsymbol{h}^H\boldsymbol{w}_{l}
\cdot s^{*}(t'+l-\tau)
\nonumber\\
&\textstyle \quad+e^{j2\pi \omega (t'+l)}\sum_{k=1}^{K}
\boldsymbol{g}_k^H\boldsymbol{w}_{l}
\cdot i_k^{*}(t'+l)\big) \label{CFO-proof-3}
\end{align}
Define $\boldsymbol{\psi}_{l}\triangleq [\psi_{l,0} \ \cdots \
\psi_{l,K}]^T$, where $\psi_{l,0}\triangleq\boldsymbol{h}^H
\boldsymbol{w}_{l}$ and $\psi_{l,k}\triangleq\boldsymbol{g}_k^H
\boldsymbol{w}_{l}$. The above set of equations can be compactly
written as
\begin{align}
\underbrace{\begin{bmatrix} \boldsymbol{Q}_0  &
\cdots & \boldsymbol{Q}_{L-1}
\end{bmatrix}}_{\triangleq \boldsymbol{Q}}
\underbrace{\begin{bmatrix}
\boldsymbol{\psi}_0 \\
\vdots \\
\boldsymbol{\psi}_{L-1}
\end{bmatrix}}_{\triangleq\boldsymbol{\psi}}=\boldsymbol{\tilde{s}},
\label{CFO-proof-4}
\end{align}
where $\boldsymbol{Q}$ is a $T\times (K+1)L$ matrix, and
\begin{align}
&\boldsymbol{Q}_{l}\triangleq
\begin{bmatrix}
\tilde{s}(1+l-\tau)    &   \tilde{i}_1(1+l)   &   \cdots    &  \tilde{i}_K(1+l)       \\
\vdots                     &   \vdots               &   \vdots    &  \vdots                   \\
\tilde{s}(T+l-\tau)    &   \tilde{i}_1(T+l)   &   \cdots    &  \tilde{i}_K(T+l)       \\
\end{bmatrix}
\end{align}
in which
\begin{align}
&\tilde{s}(t'+l-\tau)\triangleq e^{j2\pi
(\omega-\delta_f) (t'+l)} s^{*}(t'+l-\tau),
\nonumber\\
&\tilde{i}_k(t'+l)\triangleq e^{j2\pi \omega (t'+l)}
i_k^{*}(t'+l). \label{CFO-proof-5}
\end{align}
Also note that (\ref{CFO-proof-4}) can be more compactly written as
\begin{align}
\begin{bmatrix}
\boldsymbol{\tilde{s}} & \boldsymbol{Q}
\end{bmatrix}
\begin{bmatrix}
-1 \\
\boldsymbol{\psi}
\end{bmatrix}=\boldsymbol{0} \label{eqn1}
\end{align}
Recall that $s(t)$ and $\{i_k(t)\}_{k=1}^K$ are statistically
independent of each other, and each signal (including $s(t)$ and $\{i_k(t)\}_{k=1}^K$) is
a random process. For the case of $\omega\neq\delta_f$,
the matrix $[\boldsymbol{\tilde{s}} \phantom{0}\boldsymbol{Q}]$ has a
full column rank almost surely when $T>(K+1)L$. Hence there
does not exist a nonzero solution to satisfy (\ref{eqn1}).
Consequently, we cannot find a solution $\boldsymbol{w}$ to
satisfy (\ref{CFO-proof-1}), and the objective function in
(\ref{opt3}) cannot attain 0 when $\omega\neq\delta_f$.

Next, we show that if $T> (K+1)L$ and $\omega=\delta_f$, then the
solution $\boldsymbol{w}$ to (\ref{CFO-proof-1}) always belongs to
the set $\mathcal{C}$ defined in (\ref{set-C}). When
$\omega=\delta_f$, we need to examine the solution
$\boldsymbol{w}$ to the following equation:
\begin{align}
\boldsymbol{A}_{\delta_f}\boldsymbol{w}-\boldsymbol{\tilde{s}}=
\boldsymbol{0} \label{RLS-analysis-2}
\end{align}
Similar to (\ref{CFO-proof-3}) and (\ref{CFO-proof-4}), the above equation can be equivalently written as
\begin{align}
\underbrace{\begin{bmatrix} \boldsymbol{T}_0 &
\cdots & \boldsymbol{T}_{L-1}
\end{bmatrix}}_{\triangleq \boldsymbol{T}}
\underbrace{\begin{bmatrix}
\boldsymbol{\psi}_0 \\
\vdots \\
\boldsymbol{\psi}_{L-1}
\end{bmatrix}}_{\triangleq\boldsymbol{\psi}}=\boldsymbol{\tilde{s}},
\label{RLS-analysis-7}
\end{align}
where
\begin{align}
&\boldsymbol{T}_{l}\triangleq
\begin{bmatrix}
s^{*}(1+l-\tau)    &   \tilde{i}_1(1+l)   &   \cdots    &  \tilde{i}_K(1+l)       \\
\vdots             &   \vdots                 &   \vdots    &  \vdots                     \\
s^{*}(T+l-\tau)    &   \tilde{i}_1(T+l)   &   \cdots    &  \tilde{i}_K(T+l)       \\
\end{bmatrix}
\label{RLS-analysis-8}
\end{align}
in which $\tilde{i}_k(t'+l)\triangleq e^{j2\pi \delta_f (t'+l)}
i_k^{*}(t'+l)$.

Since signals $\{s(t),i_k(t),k=1,\ldots,K\}$ are statistically
independent of each other and each signal is a random process, the matrix $\boldsymbol{T}\in\mathbb{C}^{T\times
(K+1)L}$ is full column rank with probability one when $T\geq
(K+1)L$. In this case, (\ref{RLS-analysis-7}) admits a unique
solution, and it can be readily verified that this unique solution
is given by
\begin{align}
\psi_{l,0}=\begin{cases} 1& l=\tau \\ 0&
\text{otherwise}\end{cases}, \text{and} \
\psi_{l,k}= 0, \ \forall k=1,\ldots, K
\end{align}
It is clear that this unique solution corresponds to the solution
$\boldsymbol{w}$ which belongs to the set $\mathcal{C}$ defined in
(\ref{set-C}). The proof is completed here.

\bibliography{newbib}
\bibliographystyle{IEEEtran}

\end{document}